\documentclass[a4paper,UKenglish,cleveref, autoref, thm-restate]{lipics-v2021}
\nolinenumbers
\hideLIPIcs

\Crefname{section}{Sec.}{Secs.}
\Crefname{figure}{Fig.}{Figs.}

\usepackage[dvipsnames]{xcolor}

\definecolor{darkblue}{rgb}{0.0, 0.0, 0.55}
\definecolor{Darkgreen}{rgb}{0., 0.35, 0.}
\definecolor{darkgreen}{rgb}{0.0, 0.55, 0.13}
\definecolor{darkred}{rgb}{0.55, 0.0, 0.0}
\definecolor{darkviolet}{rgb}{0.58, 0.0, 0.83}
\definecolor{lightblue}{rgb}{0.68, 0.85, 0.9}
\usepackage{lstcoq}

\usepackage{stmaryrd}
\usepackage{mathtools}
\usepackage{macros}
\usepackage{tikz-cd}
\usetikzlibrary{decorations.pathmorphing}

\tikzcdset{
  squiggly/.code={%
      \pgfkeysalso{
        /tikz/decorate,
        /tikz/decoration={
          snake,
          segment length=15\pgflinewidth,
          amplitude=1.9\pgflinewidth,
          post=lineto, post length=6\pgflinewidth,
          pre=lineto, pre length=6\pgflinewidth,
          #1}}},
  arrow style=tikz,
  diagrams={>={Straight Barb[scale=0.8]}}
}

\title{Bidirectional Interpolation for the \(\lambda\)-Calculus}
\subtitle{Revisiting and Formalising Craig-\v Cubri\' c Interpolation}
\titlerunning{Bidirectional Interpolation for the \(\lambda\)-Calculus}

\author{Meven {Lennon-Bertrand}}{Université Paris Cité, INRIA, CNRS, IRIF, F-75013 Paris, France}{meven.bertrand@inria.fr}{https://orcid.org/0000-0002-7079-8826}{}

\author{Alexis Saurin}{Université Paris Cité, CNRS, INRIA, IRIF, F-75013 Paris, France}{Alexis.Saurin@irif.fr}{https://orcid.org/0009-0002-1304-5518}{}

\authorrunning{M. Lennon-Bertrand and A. Saurin}

\Copyright{Meven Lennon-Bertrand and Alexis Saurin}



\relatedversion{}\relatedversiondetails[linktext={10.4230/LIPIcs.ITP.2026.30}]{Published Version}{https://doi.org/10.4230/LIPIcs.ITP.2026.30}
\relatedversiondetails[linktext={hal.science/hal-05526634}]{Extended Version}{https://hal.science/hal-05526634}

\supplement{Source code}%optional, e.g. related research data, source code, ... hosted on a repository like zenodo, figshare, GitHub, ...
\supplementdetails[linktext={10.5281/zenodo.20308175}]{Zenodo archive}{https://doi.org/10.5281/zenodo.20308175} %linktext, cite, and subcategory are optional
\supplementdetails[linktext={github.com/MevenBertrand/interpolation}]{Github repository}{https://github.com/MevenBertrand/interpolation} %linktext, cite, and subcategory are optional

\EventEditors{Ekaterina Komendantskaya and Tobias Nipkow}
\EventNoEds{2}
\EventLongTitle{17th International Conference on Interactive Theorem Proving (ITP 2026)}
\EventShortTitle{ITP 2026}
\EventAcronym{ITP}
\EventYear{2026}
\EventDate{July 26--29, 2026}
\EventLocation{Lisbon, Portugal}
\EventLogo{}
\SeriesVolume{382}
\ArticleNo{34}
\begin{document}

\maketitle

\begin{abstract}
Craig's Interpolation theorem has a wide range of applications, from mathematical logic to computer science. Proof-theoretic techniques for establishing interpolation usually follow a method first introduced by Maehara for the Sequent Calculus and then adapted by Prawitz to Natural Deduction. 
The result can be strengthened to a proof-relevant version, taking proof terms into account: this was first established by \v Cubri\'c in the simply-typed lambda-calculus with sums and more recently in linear, classical and intuitionistic sequent calculi.
We give a new proof of \v Cubri\' c's proof-relevant interpolation theorem by building on principles of bidirectional typing, and formalise
it in \Rocq.
\end{abstract}

\section{Introduction}
\label{sec:intro}

\subparagraph*{Craig-Lyndon Interpolation}

Craig's interpolation theorem~\cite{Craig57,Craig57a} expresses a desirable, though unsurprising, property of logical entailment: a logical entailment of $B$ from $A$ can be mediated through a logical statement that only uses the \emph{common concepts} of $A$ and $B$. More precisely, a \emph{Craig Interpolant} $I$ of $A \vdash B$ is a logical formula that is built using only the predicate symbols occurring both in $A$ and $B$ and such that $A \vdash I$ and $I \vdash B$.
This was soon recognized as a central result of first-order logic~\cite{Feferman08,vanBenthem08} with application in both mathematical logic~\cite{Robinson56,Beth1953padoa,Feferman08,Vaananen08} (notably an elegant proof of Beth's definability theorem and of Robinson's joint consistency theorem) and computer science~\cite{Lavalette08,Troelstra_Schwichtenberg_2000,McMillan18} (with applications in database query planning or modularisation of model-checking), and numerous generalizations to extensions or restrictions of classical logic~\cite{Vaananen08,Schutte,Prawitz1965-PRANDA}.

Proof-theoretic methods to establish Craig's theorem essentially rely on
cut admissibility in the sequent calculus and the resulting subformula property~\cite{Maehara,Prawitz1965-PRANDA}. Indeed, this subformula property -- a property that holds for cut-free derivations, asserting that every formula appearing in the derivation is a subformula of some formula in the conclusion judgement -- allows to control the language of the interpolant. 
This provides well-structured proof methods~\cite{Maehara,Prawitz1965-PRANDA,Takeuti2013proof,Schutte,Troelstra_Schwichtenberg_2000},
and plays an important role when generalising interpolation to various logics,
as well as in finer-grained notions of interpolations, most notably Lyndon interpolation~\cite{Lyndon1959interpolation}, which is sensitive to polarity, and uniform interpolation~\cite{Pitts92},
in which a single interpolant is generated for a whole sub-language.

%\begin{itemize}
%  \item Central result in FOL \cite{Craig57,Craig57a,Lyndon1959interpolation,Maehara,Prawitz1965-PRANDA,Feferman08}
%  \item Applications in Logic \cite{Robinson56,Beth1953padoa,Feferman08,Vaananen08,vanBenthem08} and in CS \cite{Lavalette08,Troelstra_Schwichtenberg_2000,McMillan18}
%\item Factorising logical consequence
%\item Proof-theoretic methods \cite{Maehara,Prawitz1965-PRANDA,Takeuti2013proof,Schutte,Troelstra_Schwichtenberg_2000}
%\end{itemize}

\subparagraph*{Proof-relevant Interpolation}
Considering the importance of proof-theoretic methods for interpolation, it is not so surprising that interpolation theorems can be given a proof-relevant refinement, as recently advocated by Saurin~\cite{Saurin25fscd}. This means that not only can one factor the entailment through the interpolant formula, but also that
the proofs from and to the interpolant, once composed, yield the original proof.
Even more interesting is that proof-relevant interpolation was implicit already
in Maehara's original results in the 60s.

This is exactly what is at stake in \v Cubri\'c's interpolation theorem~\cite{Cubricphd,Cubric1994} for the simply typed \(\lambda\)-calculus with sums and empty types (\stlcp), a proof-relevant refinement of Prawitz's proof of interpolation~\cite{Prawitz1965-PRANDA} for natural deduction. 
This theorem essentially states that for every types $A,B$ and every normal term $t$
such that $\typing{x: A}{t}[B]$, there is a type $C$, built from the atomic types occurring both in $A$ and $B$, and terms $\typing{x: A}{u}[C]$ and $\typing{y: C}{v}[B]$ such that 
$\red{(\lambda y. v)~u}{t}$.
Owing to the fact that \stlcp is the internal language
of bi-cartesian closed categories (biCCC), \v Cubri\'c moreover
leverages his interpolation result to obtain categorical corollaries,
including categorifications of Pitt's results on Heyting algebras, the
posetal counterpart to biCCCs.

\subparagraph*{Interpolation meets bidirectional typing}
Although very interesting, \Cubric's original proof is… not exactly pretty.
The core argument goes by an induction on the size of terms, gradually replacing
subterms by variables. Moreover, the one case dismissed in the proof as
``the same as the previous'' \cite[Subcase 6.3.2]{Cubric1994}, is in fact not
at all similar.
However, we can discern more structure in the proof, as it essentially follows
a characterisation of normal forms in \stlcp going back to
Prawitz. We have made progress since Prawitz,
and it is nowadays routine \cite{Abel2013,Fiore2022}
to re-cast this characterisation as an explicit,
inductive description of normal forms, mutually with neutrals,
those normal forms consisting of stuck eliminations.
Our new proof thus goes by a direct, mutual induction
on this inductive characterisation. 

A difficulty is that, in the presence of positive type formers
(sums and the empty type), \(\beta\)-normal forms are not enough: interpolation
works well for terms/proofs satisfying the subformula property, which in
\stlcp mere \(\beta\)-normal forms do not have. We hence need to go further and
normalise with respect to \emph{commuting conversions} \cite{Girard1989}, which
exchange certain eliminations to recover a class of normal forms with a
good subformula property.

Surprisingly, our inductive normal forms are
naturally interpreted in terms of bidirectional typing \cite{Pierce2000,Dunfield2021}.
Although this connection has been drawn before, it is
often believed that it breaks for sums. We think this is
a confusion stemming from the fact that \stlcp has two notions of normal forms.
The first is characterised as having the subformula property, which
enables our interpolation proof. But, contrarily to the situation without sums,
these do not give unique representative of equivalence
classes for the equational theory of biCCCs, for which one needs
a second, stronger notion of normal forms \cite{Scherer2017}. Since bidirectional typing
has all to do with information flow, it naturally coincides with the
\emph{first} kind of normal forms, although it does not correspond with the second
any more. We believe this relation between bidirectionalism and
the subformula property deserves more publicity, which we hope this paper can
contribute towards.

\subparagraph*{Contributions}
We develop and formalise, in the \Rocq proof assistant,
the meta-theory of the \(\lambda\)-calculus with sums,
including the first formalised proof of normalisation
for a calculus with commuting conversions; we relate normal
forms to bidirectional typing, and
use this in a new proof of proof-relevant interpolation.

\subparagraph*{Organisation}
The paper is organised as follows.
In \Cref{sec:stlc-sums} we recall the definition and basic properties of the simply typed  \(\lambda\)-calculus with sums, with a focus on commuting conversions.
\Cref{sec:bidirectional} introduces a bidirectional type system and shows that it characterises the normal forms of \stlcp; on the way, we establish normalisation of \stlcp to strengthen the interpolation result that we establish in \Cref{sec:interpolation}, where we prove a proof-relevant Lyndon-style interpolation theorem, first for the constant-free calculus
then extended to handle constants and equations.
We discuss formalisation choices in \Cref{sec:formalisation} and related work in 
\Cref{sec:related}, and conclude in \Cref{sec:conclusion}.

\subparagraph*{Notations} We work in dependent type theory, but
to avoid confusion between judgments of the ambiant logic and those of
the object language,
we denote meta-level typing as \(\mty\) and the meta-level universes
of propositions as \(\Prop\).
Colours are used to suggest information:
a term being \textcolor{ccol}{normal} or \textcolor{icol}{neutral},
or an object containing \textcolor{scol}{source},
\textcolor{tcol}{target} or \textcolor{mcol}{shared} vocabulary.
No information is conveyed by colour alone, it is only used as an extra visual
cue. Accordingly, the same name with different colours denotes
\emph{the same object} (with updated knowledge).
% For instance, when we learn that a source type
% only contains target vocabulary, it is now a valid interpolant
% and thus changes colour.
Throughout the text, \formfile[]{}[] indicates links to the code.

\section{The simply-typed \(\lambda\)-calculus with sums}
\label{sec:stlc-sums}

\subsection{Syntax}

\subparagraph*{Types and terms}

\begin{figure}
\begin{mathpar}
  \mprset{andskip=1em plus 0.5fil minus 0.5em}
  \vspace*{-1.5em}
  \jform{\ctxty[\base]{\Gamma}}[Contexts \(\Gamma,\Delta,\dots \in \Ctx_{\base}\)]
  \inferdef{ }{\ctxty{\emp}} \and
  \inferdef{\ctxty{\Gamma} \\ \ctxty{A}}{\ctxty{\Gamma . (x : A)}} \vspace*{-.5em}\\
  \vspace*{-.7em}
  \jform{\ctxty[\base]{A}}[Types \(A,B,\dots \in \Ty_{\base}\)]
  \inferdef{b \in \base}{\ctxty[\base]{b}} \and
  \inferdef{ }{\ctxty{\top}} \and
  \inferdef{\ctxty{A} \\ \ctxty{B}}{\ctxty{A \times B}} \and
  \inferdef{\ctxty{A} \\ \ctxty{B}}{\ctxty{A \to B}} \and
  \inferdef{ }{\ctxty{\bot}} \and
  \inferdef{\ctxty{A} \\ \ctxty{B}}{\ctxty{A + B}} \vspace*{-.5em}\\
  \vspace*{-.7em}
  \jform{\typing[\lang]{\Gamma}{t}[A]}[Terms \(t,u,\dots \in \Tm_{\lang}(\Gamma,A)\)]
  \inferdef[Cst]{c \in \const}{\typing[(\base,\const,\tyof)]{\Gamma}{c}[\tyof(c)]} \and
  \inferdef[Var]{(x : A) \in \Gamma}{\typing{\Gamma}{x}[A]} \and
  \inferdef[Star]{ }{\typing{\Gamma}{\st}[\top]} \and
  \inferdef[Pair]{\typing{\Gamma}{t}[A] \\ \typing{\Gamma}{u}[B]}{
      \typing{\Gamma}{(t,u)}[A \times B]} \and
  \inferdef[Proj]{\typing{\Gamma}{p}[A_1 \times A_2]}{
      \typing{\Gamma}{\proj{i}{p}}[A_i]} \and
  % \inferdef[Proj2]{\typing{\Gamma}{p}[A \times B]}{
  %     \typing{\Gamma}{\proj{2}{p}}[B]} \and
  \inferdef[Lam]{\typing{\Gamma.(x : A)}{t}[B]}{
      \typing{\Gamma}{\l x.t}[A \to B]} \and
  \inferdef[App]{\typing{\Gamma}{t}[A \to B] \\ \typing{\Gamma}{u}[A]}{
      \typing{\Gamma}{t~u}[B]} \and
  \inferdef[Raise]{\typing{\Gamma}{t}[\bot]}{\typing{\Gamma}{\rai t}[A]} \and
  \inferdef[InL]{\typing{\Gamma}{a}[A]}{\typing{\Gamma}{\inl a}[A + B]} \and
  \inferdef[InR]{\typing{\Gamma}{b}[B]}{\typing{\Gamma}{\inr b}[A + B]} \and
  \inferdef[If]{
    \typing{\Gamma}{s}[A + B] \\
    \typing{\Gamma.(x:A)}{b_l}[T] \\
    \typing{\Gamma.(x:B)}{b_r}[T] \\
  }{\typing{\Gamma}{\ite{s}{b_l}{b_r}}[T]}
\end{mathpar}
\caption{Types and terms of \stlcp with respect to a language
  \(\lang = (\base,\const,\tyof)\)}
\label{fig:stlc}
\end{figure}

We start with the simply-typed \(\lambda\)-calculus,
abbreviated as \stlcp,
with function types \(\to\), product types \(\times\),
sum types \(+\), and unit and empty types \(\top\) and \(\bot\).
To be able to state an interesting version of interpolation, we also
include an unspecified set of base types.
To make the result more interesting,
we likewise include a language of constants.

\begin{definition}[\formline{BasicAst.v}{18}{Lang} Language]
  \label{def:lang}
  A \emph{language} \(\lang\) consists of:
  \begin{itemize}
    \item a set \(\base\) of base types;
    \item a set \(\const\) of constants;
    \item types for each constant, \ie
      a function \(\tyof \in \const \to \Ty_{\base}\).%
    \footnote{\(\Ty_{\base}\) is the set of types generated by
    \(\base\), to be defined in \cref{def:types-terms}; formally these
    two definitions should be intertwined.} 
  \end{itemize}
  % We write \(\emptyset\) for the empty language, \ie that with no base types nor
  % constants.
\end{definition}

\begin{definition}[\formline{Typing.v}{16}{has_type} Types and terms]
\label{def:types-terms}
The definition of types and terms is given in \cref{fig:stlc}.
We decompose types into the \emph{negative} types (\(\to\), \(\times\) and
\(\top\)) and the \emph{positive} types (base types, \(+\) and \(\bot\)).
If \(\lang = (\base,\const,\tyof)\) we write
\(\Ty_{\lang}\) for \(\Ty_{\base}\), and similarly for contexts.
When it is not particularly relevant, we omit the language subscript.
\end{definition}

\subparagraph*{Substitutions}

\begin{figure}
\begin{mathpar}
  \jform{\typing{\Delta}{\sigma}[\Gamma]}[Substitutions
    \(\sigma,\tau,\dots \in \Sub_{\lang}(\Delta,\Gamma)\)]
  \inferdef[Id]{ }{\typing{\Gamma}{\id}[\Gamma]} \and
  \inferdef[Comp]{\typing{\Gamma}{\sigma}[\Delta] \and \typing{\Delta}{\tau}[\Theta]}{
      \typing{\Gamma}{\sigma \pcomp \tau}[\Theta]} \and
  \inferdef[Wk]{ }{\typing{\Gamma.(x:A)}{\wk}[\Gamma]} \and
  \inferdef[Cons]{\typing{\Delta}{\sigma}[\Gamma] \\ \typing{\Delta}{t}[A]}{
      \typing{\Delta}{\sigma,t}[\Gamma.(x:A)]} \and
  \inferdef[Up]{\typing{\Delta}{\sigma}[\Gamma]}{\typing{\Delta.(x:A)}{\Up \sigma}[\Gamma.(x:A)]} \and
  \inferdef[One]{\typing{\Gamma}{t}[A]}{\typing{\Gamma}{t..}[\Gamma.(x:A)]}
\end{mathpar}
\caption{\formline{Typing.v}{164}{TypingSubst} Substitutions and their typing}
\label{fig:subst}
\end{figure}

We work with parallel substitutions, the action of which we denote
\(\subs{t}{\sigma}\).
Although we use variable names for exposition,
we formally think of them in the de Bruijn style \cite{deBruijn1972}, and
formalise them this way in \Rocq. Accordingly, the
substitution calculus we use is essentially name-free.
We include the operations on substitutions we use in the rest
of the text in \cref{fig:subst}.%
\footnote{
  Although we describe the operations using rules, substitutions are not
  inductively generated by these!
  Rather, these operations should be seen as an ``API'' for substitutions,
  whose actual representation is irrelevant.
  In the formalisation, they are represented as functions \texttt{nat -> term},
  but they could very well be represented as lists without any major change.
}
The typing judgment \(\typing{\Delta}{\sigma}[\Gamma]\)
means that \(\sigma\) takes a variable \((x : A) \in \Gamma\) to a
term \(\sigma(x) \in \Tm(\Delta,A)\). The primitive's types should already be
informative, but let us describe them succinctly:
\(\id\) is the identity, mapping each variable to itself;
\(\sigma \pcomp \tau\) is composition,
mapping \(x\) to \(\subs{\sigma(x)}{\tau}\);
\(\wk\) drops the last variable in context, and maps all other variables to themselves;%
\footnote{In named syntax this is a no-op, but we include it
as it changes the context in which a term is considered.}
\(\sigma,t\) extends
a substitution \(\sigma\) by mapping the last variable in context to \(t\);
\(\Up \sigma\) maps the last variable to itself, and acts as \(\sigma\) otherwise;
finally \(t..\) maps the last variable to \(t\) and
acts as the identity otherwise.
The following substitution rule is admissible
(\formline{Typing.v}{312}{subst_typing}):
\begin{mathpar}
  \inferdef[Subst]{\typing{\Delta}{\sigma}[\Gamma] \\ \typing{\Gamma}{\sigma}[A]}
    {\typing{\Delta}{\subs{t}{\sigma}}[A]}
\end{mathpar}

\emph{Renamings}, denoted \(\Ren(\Delta,\Gamma)\), are the substitutions that map
variables to variables.

\subsection{Reduction and equations}
\label{sec:red-eq}

The language defined in \cref{fig:stlc}, quotiented by the adequate equational
theory, is the internal language of bi-cartesian closed categories (biCCC), \ie
with finite products, finite coproducts, and exponentials. Normalisation
for this equational theory is unwieldy due to \(\eta\) laws for positive types
(see Scherer~\cite{Scherer2017}). Fortunately,
we can get away with normalisation for a less powerful reduction
containing only the rules known as \emph{commuting conversions} \cite{Girard1989}.
This is because normal forms for this reduction have a good enough subformula
property, even though they are not unique representatives for the equational
theory of biCCCs.

\subparagraph*{Reduction}

To express commuting conversions, we use \emph{eliminations},
given by the following grammar (\formline{Ast.sig}{21}{}):
\( e \Coloneqq \sapp{t} \mid \sproj{i} \mid \site{t}{t}\)
The \emph{zip} of a term against an elimination \(\plug{e}{t}\) is defined in the obvious way (\formline{Elim.v}{9}{zip}).

Reductions are of two kinds: the usual \(\beta\)-reductions, corresponding
to computation --~applied abstraction, projection of a pair,
and branching on a constructor~-- and commuting conversions, which push
an elimination across another one of a positive type (\ie \(\bot\) or
a sum). In the case of \(\bot\), the elimination gets completely removed, as it
is unnecessary: a proof of \(\bot\) can be used to inhabit
any type, there is no need to further eliminate it. In the case of the sum,
the elimination gets duplicated in each branch of the \(\iteop\).
\begin{definition}[\formline{Reduction.v}{121}{term_red} Reductions]
  \emph{One-step reduction} \(\oredop\)
  is the congruent closure of the rules given in \cref{fig:stlc-red}.
  \emph{Reduction} \(\redop\) is the reflexive, transitive closure of one-step reduction.
\end{definition}

\begin{figure}
\begin{mathpar}
  % \jform{\conv{\Gamma}{t}{u}[A]}[Equality of terms]
  % \inferdef[Refl]{\typing{\Gamma}{t}[A]}{\conv{\Gamma}{t}{t}[A]} \and
  % \inferdef[Sym]{\conv{\Gamma}{t}{u}[A]}{\conv{\Gamma}{u}{t}[A]} \and
  % \inferdef[TRans]{\conv{\Gamma}{t}{u}[A] \\ \conv{\Gamma}{t}{v}[A]}
  %   {\conv{\Gamma}{t}{v}[A]} \\
  % \text{congruence/substitution rule} \\
  % \inferdef[Eq]{(t,u) \in \mathcal{E}(A)}{\conv{\Gamma}{t}{u}[A]} \and
  % \inferdef[Unit\(\eta\)]{\typing{\Gamma}{t}[\top]}{\conv{\Gamma}{t}{\st}[\top]} \\
  % \inferdef[Pair\(\beta\)]{\typing{\Gamma}{t}[A] \\ \typing{\Gamma}{u}[B]}{
  %   \conv{\Gamma}{\proj{1}{(t,u)}}{t}[A] \and
  %   \conv{\Gamma}{\proj{2}{(t,u)}}{u}[A]
  % } \and
  % \inferdef[Pair\(\eta\)]{\typing{\Gamma}{p}[A \times B]}{
  %   \conv{\Gamma}{p}{(\proj{1}{p},\proj{2}{p})}[A \times B]} \\
  % \inferdef[Fun\(\beta\)]{\typing{\Gamma,x:A}{t}[B] \\ \typing{\Gamma}{u}[A]}{
  %   \conv{\Gamma}{(\l x.t)~u}{\subs{t}{x}{u}}[A]} \and
  % \inferdef[Fun\(\eta\)]{\typing{\Gamma}{f}[A \to B]}{
  %   \conv{\Gamma}{f}{\l x.(f~x)}[A \to B]}
  \jform{\ored{t}{u}}[One-step reduction]
  \inferdef[Pair\textsubscript{1}\(\beta\)]{ }{\ored{\proj{1}{(t,u)}}{t}} \and
  \inferdef[Pair\textsubscript{2}\(\beta\)]{ }{\ored{\proj{2}{(t,u)}}{u}} \and
  \inferdef[Fun\(\beta\)]{ }{\ored{(\l x.t)~u}{\subs{t}{u..}}} \and
  \inferdef[InL\(\beta\)]{ }{\ored{\ite{\inl a}{b_l}{b_r}}{\subs{b_l}{a..}}} \and
  \inferdef[Inr\(\beta\)]{ }{\ored{\ite{\inr b}{b_l}{b_r}}{\subs{b_r}{b..}}} \\
  \inferdef[CommRais]{ }{\ored{\plug{e}{\rai t}}{\rai t}} \and
  \inferdef[CommIf]{ }{\ored{\plug{e}{\ite{s}{b_l}{b_r}}}{\ite{s}{\plug{e}{b_l}}{\plug{e}{b_r}}}}
\end{mathpar}
\caption{\formline{Reduction.v}{50}{term_ored} Reduction (congruence rules omitted)}
\label{fig:stlc-red}
\end{figure}

\begin{theorem}[\formline{ReductionConfluence.v}{173}{confluence} Confluence]
\label{thm:confluence}
Reduction is confluent: if \(\red{t}{u}\) and \(\red{t}{u'}\), then there exists
\(v\) such that \(\red{u}{v}\) and \(\red{u'}{v}\).
\end{theorem}

Since confluence is rather standard \cite{girard1987proof} yet tedious to prove, and
orthogonal to our focus, we do not formally prove it in \Rocq,
but only provide a pen-and-paper proof sketch.%
\footnote{Surprisingly, despite our best efforts we were not able to find a
detailed proof of that fact.}
% although Girard repeatedly claims it is trivial
% \cite{girard1987proof,Girard1989} with the Tait--Martin-Löf method, although
% it does actually not apply here.}
% (See details in Appendix). \as{Mettre cela en place!}

% \begin{proof}[Proof sketch]
% By Hindley-Rosen lemma, it suffices to show that 1. \(\beta\)-reduction is confluent 2. commuting conversion is confluent and 3. the two commute. 
% For the second, strong normalisation
%   of commuting conversion can be shown by a simple \(\mathbb{N}\)-valued measure
%   (\formline{ReductionConfluence.v}{114}{strong_nor}),
%   so showing local confluence is enough by Newman's lemma (\formline{ReductionConfluence.v}{148}{confluent}).
%   This is painful due to the many critical
%   pairs, so we have checked this with the automated tool \textsc{CSI\textsuperscript{ho}}
%   \cite{Nagele2017}.
% \end{proof}

\begin{proof}
  We use the Hindley-Rosen lemma \cite[Prop. 3.3.5]{Barendregt1985}, by which
  it suffices to show that 1. \(\beta\)-reduction is confluent 2. commuting
  conversion is confluent 3. the two commute. For the first, since there is
  no critical pair Takahashi's variant of parallel reduction à la Tait--Martin-Löf
  \cite{Tait1967,Takahashi1995} does the trick. For the second, strong normalisation
  of commuting conversion can be shown by a simple \(\mathbb{N}\)-valued measure
  (\formline{ReductionConfluence.v}{114}{strong_nor}),
  so by Newman's lemma \cite[3.1.42]{Barendregt1985}, it
  suffices to show local confluence (\formline{ReductionConfluence.v}{148}{confluent}).
  This is painful due to the many critical
  pairs, so we have checked this with the automated tool \textsc{CSI\textsuperscript{ho}}
  \cite{Nagele2017}. Finally, commutation again does not need normalisation because
  critical pairs between \(\beta\) and commuting conversion
  are all akin to the following, which can be closed in one step.
  \begin{tikzcd}[sep=small]
    & \ite{(\inl{a})}{b_l}{b_r}~u
      \arrow[squiggly,rd]
      \arrow[squiggly,ld] & \\
    \ite{(\inl{a})}{b_l~\subs{u}{\wk}}{b_r~\subs{u}{\wk}}
      \arrow[squiggly,rd]
      && \subs{b_l}{a..}~u \arrow[ld,equal] \\
    & \subs{b_l}{a..}~u
  \end{tikzcd}
\end{proof}

\subparagraph*{Equational theories}

To formally relate reduction to biCCCs, we can define
(\formline{Equations.v}{29}{conversion}) a notion of conversion
\(\conv{\Gamma}{t}{t'}[A]\) capturing equality of morphisms in a biCCC
as the least congruent equivalence relation containing
\(\beta\) and \(\eta\) laws, encoding the
universal property of each type former.
These are strictly stronger than
the commuting conversions of \cref{fig:stlc-red}; indeed we have
\[\conv
  {x : A + A, y : A + A, z : A}
  {\ite{x}{\ite{y}{z}{z}}{\ite{y}{z}{z}}}
  {\ite{y}{\ite{x}{z}{z}}{\ite{x}{z}{z}}}
  [A]
\]
by \(\eta\) for sums, but the two terms are not related by commuting conversions
alone.

This is not an issue, as our interpolation proof only uses conversion
in covariant positions, so we show stronger theorems using reduction,
and can weaken them to conversion, owing to:
\begin{lemma}[\formline{Equations.v}{720}{subject_reduction} Subject reduction]
  If \(\typing{\Gamma}{t}[A]\) and \(\red{t}{u}\) then \(\conv{\Gamma}{t}{u}[A]\).
\end{lemma}
We actually go further and define a notion of \emph{theory}
(\formline{Equations.v}{7}{Theory}),
which extends the specification of a language with
a set of arbitrary equations that are added to its conversion.
The proofs are essentially unchanged, but this means that
our interpolation can be leveraged to prove general properties of
biCCCs by viewing \stlc as their internal language,
although we do not explore this further here, and refer to
\Cubric \cite{Cubric1994,Cubricphd} for details.

\section{Bidirectional typing, normal forms and the subformula property}
\label{sec:bidirectional}

In logic, a key characteristic of derivations is the \emph{subformula property}, satisfied
when every formula appearing in the derivation is a subformula of some formula
in the conclusion judgment --~hypothesis or conclusion. This
does not hold of arbitrary derivations, but is satisfied by those which
are sufficiently normalised, \ie cut-free proofs \cite{Girard1989}.
In a derivation satisfying the subformula property, no formula is ever ``invented'',
a powerful structural property with many consequences. Interpolation is
for instance typically proven for cut-free derivations.
One then obtains properties for all provable
statements by a combination of cut elimination, by which every
proof is equal to a cut-free one, and an analysis of cut-free proofs, which can
exploit its more constrained structure, including the subformula property.

For a type theorist, this might seem a bit mysterious. Cut elimination, by
the Curry-Howard correspondence, is well-known to coincide with reduction of
terms viewed as proof witnesses. However, what does the subformula property
amount to? This is much less known, but it can be very helpfully linked to
\emph{bidirectional typing}.

In bidirectional typing,
which originates from the implementation of type-checking algorithms \cite{Coquand1996,Pierce2000},
typing is decomposed into the mutually defined
\emph{inference} and \emph{checking} judgments. The operational intuition is
that in inference the type is an output,%
\footnote{\textit{I.e.} the question is ``What is the type of this term?''.}
while in checking it is an input.%
\footnote{And the question is ``Does this term have this type?''.}
This attention to the flow of information is, in a sense, what
the subformula property also intends to capture: in both cases, the intuition
is that type information always comes from somewhere and is never invented.
% And in fact bidirectional typing can helpfully be leveraged to exploit
% this subformula property.

With this view it is perhaps less surprising that bidirectional
typing can be used to capture normal forms, since
these correspond to cut-free proofs which have a good
subformula property. This remark, however, has powerful consequences:
because bidirectional type systems have a nice, inductive structure, they
give a good basis on which to build clean proofs of theorems that rely on
the subformula property, such as interpolation. The remainder of this article
is dedicated to delivering that idea: the rest of this section sets up the
bidirectional characterisation of normal forms, while \cref{sec:interpolation}
derives the interpolation result.

\subsection{A bidirectional characterisation of normal forms}

\begin{figure}
\begin{mathpar}
  \jform{\check{\Gamma}{t}{A}}[Checking/normal forms \(\ccol{t},\ccol{u},\dots \in \Nf_{\lang}(\Gamma,A)\)]
  
  \inferdef[Star]{ }{\check{\Gamma}{\st}{\top}}
  \label{rule:star-b} \and
  \inferdef[Pair]{\check{\Gamma}{t}{A} \\ \check{\Gamma}{u}{B}}{
      \check{\Gamma}{(t,u)}{A \times B}}
  \label{rule:pair-b} \and
  \inferdef[Lam]{\check{\Gamma,x : A}{t}{B}}{
      \check{\Gamma}{\l x.t}{A \to B}}
  \label{rule:lam-b} \and
  \inferdef[InL]{\check{\Gamma}{a}{A}}{
      \check{\Gamma}{\inl a}{A + B}}
  \label{rule:inl-b} \and
  \inferdef[InR]{\check{\Gamma}{b}{B}}{
      \check{\Gamma}{\inr b}{A + B}}
  \label{rule:inr-b} \and
  \inferdef[Demote]{\infer{\Gamma}{t}{A} \\ A = B}{\check{\Gamma}{\dem{\icol{t}}}{B}}
  \label{rule:demote-b} \\
  \inferdef[Raise]{\infer{\Gamma}{\icol{t}}{\bot}}{\check{\Gamma}{\rai \icol{t}}{A}}
  \label{rule:raise-b}\and
  \inferdef[If]{
    \infer{\Gamma}{s}{A + B} \\
    \check{\Gamma.(x:A)}{b_l}{T} \\
    \check{\Gamma.(x:B)}{b_r}{T} \\
  }{\check{\Gamma}{\ite{\icol{s}}{b_l}{b_r}}{T}}
  \label{rule:if-b} \\
  \jform{\infer{\Gamma}{t}{A}}[Infering/neutral forms \(\icol{t},\icol{u},\dots \in \Ne_{\lang}(\Gamma,A)\)]
  \inferdef[Cst]{c \in \mathcal{C}}{\infer[\lang]{\Gamma}{c}{\tyof(c)}}
  \label{rule:cst-b} \and
  \inferdef[Var]{(x : A) \in \Gamma}{\infer{\Gamma}{x}{A}}
  \label{rule:var-b} \\
%  \inferdef[Proj]{\infer{\Gamma}{p}{A \times B}}{
%      \infer{\Gamma}{\proj{1}{p}}{A} \and \infer{\Gamma}{\proj{2}{p}}{B}} \and
  \inferdef[Proj$_1$]{\infer{\Gamma}{p}{A \times B}}{
      \infer{\Gamma}{\proj{1}{p}}{A}} \and
        \inferdef[Proj$_2$]{\infer{\Gamma}{p}{A \times B}}{
      \infer{\Gamma}{\proj{2}{p}}{B}} 
        \label{rule:proji-b} \and
  \inferdef[App]{\infer{\Gamma}{t}{A \to B} \\ \check{\Gamma}{u}{A}}{
      \infer{\Gamma}{t~\ccol{u}}{B}}
  \label{rule:app-b}
\end{mathpar}
\caption{\formline{Bidir.v}{11}{normal} Bidirectional typing for the \stlcp,
  with respect to a language \(\lang = (\base,\const,\tyof)\)}
\label{fig:bidir}

\end{figure}

Our bidirectional type system is given in \cref{fig:bidir}. It is
inspired by the normal forms of Abel and Sattler~\cite{Abel2019a},
modified to not enforce \(\eta\)-long forms. It can be interpreted in two ways.

In the bidirectional mindset, rules can be read as
a typing algorithm for a fragment of \stlcp. Rules for the inference judgment
\(\infer{\Gamma}{t}{A}\) correspond to settings where it is feasible to infer
the type \(A\) from \(t\), by inferring that of a carefully chosen subterm
(in the case of a negative elimination \nameref{rule:app-b}, \nameref{rule:proji-b})
or directly reading it from the
local (\nameref{rule:var-b}) or global (\nameref{rule:cst-b}) context.
Conversely, rules for checking \(\check{\Gamma}{t}{A}\) are those where it is better to check
that a term has a given type. This encompasses constructors
(\nameref{rule:star-b}, \nameref{rule:pair-b}, \nameref{rule:lam-b}, \nameref{rule:inl-b},
\nameref{rule:inr-b}), where the type is decomposed into parts
which are used to check the subterms;%
\footnote{One could wonder why \nameref{rule:star-b}
is not inferring. While possible, it is more uniform to treat \(\star\)
as a nullary version of pairing. Anyway, since it can
only appear in a checking position this makes little difference.}
any inferring term, which we can check provided the inferred and checked types
coincide (\nameref{rule:demote-b}); and more interestingly, the
elimination for positive types (\nameref{rule:raise-b}, \nameref{rule:if-b}).
It is rather clear that \(\rai\) cannot be reasonably inferring: how are
we supposed to guess its type? For \(\iteop\), we could infer
the type of the branches, comparing them, and if they agree infer that type
for the whole pattern-matching. However, our version ends up being better behaved.

In this view, the bidirectional type system arises from the will to
type annotation-free terms. To do so, we can only rely on type information
which is already available, either in the context, or in the type
in the case of checking. This means that all the type information
should appear in these, which, to reiterate,
is exactly what the subformula property captures!

The other approach views the \(\check{\Gamma}{t}{A}\)
judgment as characterising \emph{normal forms}, those which do not reduce.
When inductively characterising normal forms, it is necessary to mutually define
\emph{neutral forms}, those which can be used as the main argument of
an elimination without creating redexes. Usually, any elimination of a neutral
is again a neutral. But with commuting conversions this will not do,
as an eliminator on top of \(\rai\) or \(\iteop\) reduces. So in that view too,
it makes sense to classify these as normal forms but not neutrals.

We prove that this adequately captures well-typed terms that do not reduce,
\ie that our inductive characterisation indeed captures normal forms in the
more traditional sense.
% Here we express the theorems only for normal forms,
% but they of course also hold for neutrals.

\begin{lemma}[\formline{MetaTheory.v}{56}{normal_bidir} Normality]
  If \(\check{\Gamma}{t}{A}\), then \(\typing{\Gamma}{t}[A]\)
  and there is no \(t'\) such that \(\ored{t}{t'}\).
\end{lemma}

\begin{theorem}[\formline{MetaTheory.v}{29}{progress} Progress]
  If \(\typing{\Gamma}{t}[A]\), then either there is \(t'\) such that
  \(\ored{t}{t'}\), or \(\check{\Gamma}{t}{A}\).
\end{theorem}

\subsection[Normalisation for STLCp]{Normalisation for \stlcp}

In \cref{sec:interpolation}, we prove interpolation for normal forms only.
To apply it more generally, we show that every well-typed term
reduces to a normal form. % \ie a normalisation theorem.
We give a proof by
logical relations, heavily inspired by Abel and Sattler~\cite{Abel2019a},
although compared to their normalisation-by-evaluation approach we are able to
obtain a result about reduction rather than conversion.
% As for all logical relation argument, the proof goes by defining a semantic
% notion of typing, which implies normalisation, but is strengthened to be stable
% under the semantic version of all typing rules.

% \subparagraph*{Definitions}

In a logical relation argument, the key definition is that of values,
which goes by induction on types.
As is standard, values at negative
types are normal forms whose observations are reducible:%
\footnote{Reducible terms are those which reduce to values.}
a value at a pair type is one with reducible projections, and a value
at function type takes value inputs to reducible outputs.

However, things are trickier at positive types. Usually, one
would say that a value at a positive type is given either by a constructor
or a neutral. This will not do here, as we must take into account eliminations
of \(+\) and \(\bot\), which are not neutrals. Instead, normal forms at positive
types are case trees, with nodes branching on neutrals of positive
types, and leaves having the usual shape of a neutral or a constructor.
To capture this, we introduce a notion of \emph{covering}, which
more generally describes terms with a case tree structure, and
echoes the use of the free cover monad in \textit{op.\ cit.}~\cite{Abel2019a}.

\begin{definition}[\formline{MetaTheory.v}{97}{covered} Covering]
  A term \(t \in \Tm(\Delta,A)\) is \emph{covered} by a predicate
    \(F \in (\Gamma \in \Ctx) \to \Tm(\Gamma,A) \to \Prop\)
  if either:
  \begin{itemize}
    \item \(F(t)\) holds;
    \item \(t = \rai n\) with \(\infer{\Gamma}{n}{\bot}\);
    \item \(t = \ite{n}{b_l}{b_r}\), with \(\infer{\Gamma}{n}{A_1 + A_2}\),
      and \(b_l\), \(b_r\) are covered, respectively in \(\Gamma.(x:A_2)\) and
      \(\Gamma.(x:A_2)\).
  \end{itemize}
  We say that a term \(t \in \Tm(\Gamma,A)\) is ``covered by neutrals'' if it is covered
  by the predicate \(F(\Delta,u) \Coloneqq \infer{\Delta}{u}{A}\), and similarly
  for other predicates on terms.
\end{definition}

We can view this as lifting a base predicate \(F\) to a
new one which incorporates the case tree structure.
Indeed, given any predicate \(F \in (\Gamma \in \Ctx) \to \Tm(\Gamma,A) \to \Prop\),
``being covered by \(F\)'' is again a predicate of the same type.

% The next definition of values and reducible terms goes by induction
% on types.

\begin{definition}[\formline{MetaTheory.v}{135}{value} Values]
  A \emph{raw value} at a positive type \(P\) is a \(t \in \Tm(\Gamma,T)\)
  such that:
  \begin{itemize}
    \item \(P\) is a base type and \(t\) is neutral;
    \item \(P\) is \(\bot\) and \(t\) is neutral;
    \item \(P\) is \(A + B\) and \(t\) is either neutral, \(\inl(a)\) where
      \(a\) is a value of type \(A\) in \(\Gamma\), or \(\inr(b)\) where \(b\) is
      a value of type \(B\) in \(\Gamma\).
  \end{itemize}

  A \emph{value} at type \(T\) is a term \(v \in \Tm(\Gamma,T)\) such that:
  \begin{itemize}
    \item if \(T\) is a positive type, \(v\) is covered by raw values;
    \item if \(T\) is \(\top\), \(v\) is a normal form;
    \item if \(T\) is \(A_1 \times A_2\), \(v\) is a normal form and \(\proj{i}{v}\)
      is reducible at \(A_i\);
    \item if \(T\) is \(A \to B\), \(v\) is a normal form and for any
      context \(\Delta\), renaming \(\typing{\Delta}{\rho}[\Gamma]\),
      and value \(u \in \Tm(\Delta,A)\), \(\subs{v}{\rho}~u\) is reducible.
  \end{itemize}
\end{definition}
  
\begin{definition}[\formline{MetaTheory.v}{128}{reducible} Reducibility]
  A term \(t\) is \emph{reducible} if \(\lred{t}{v}\) for some value \(v\).
\end{definition}

This extends to substitutions: a value substitution
is a substitution \(\sigma \in \Sub(\Delta,\Gamma)\) mapping each variable
\((x : A) \in \Gamma\) to a value \(\sigma(x)\) at \(A\).
We use this to define semantic typing.

\begin{definition}[\formline{MetaTheory.v}{160}{sem_has_type} Semantic typing]
  A term \(t\) is \emph{semantically well-typed}, written \(\semty{\Gamma}{t}[A]\),
  if for any value substitution \(\sigma \in \Sub(\Delta,\Gamma)\), \(\subs{t}{\sigma}\)
  is reducible at \(A\) in \(\Delta\).
\end{definition}

% \subparagraph*{Properties}

Before we can get on to proving the semantic counterparts to the typing rules in
\cref{fig:stlc}, we first need a handful of properties of covers, values and
reducible terms.

\begin{lemma}[\formline{MetaTheory.v}{193}{reify} Reification]
  Any value is a normal form.
\end{lemma}

\begin{lemma}[\formline{MetaTheory.v}{202}{reflect} Reflection]
  \label{lem:reflect}
  Any neutral is a value. 
\end{lemma}
These lemmas are named by analogy with the reify/reflect functions of normalisation by
evaluation. The lemmas are both proven by induction on
types. For the first, the key ingredient is that a term covered by normal
forms is again a normal form. We need reification in reflection, to handle
reducibility at function types where to conclude that applying the neutral
to a value yields a neutral we need to turn that value into a normal form.

Next, we show that a term covered by reducible terms is reducible.
At positive types, it follows from the next lemma, which can be read as
the multiplication of a covering monad.
\begin{lemma}[\formline{MetaTheory.v}{112}{join_covered} Multiplication]
  A term covered by terms covered by \(F\) is covered by \(F\).
\end{lemma}

At negative types, however, things are not that simple. Consider the function type:
we assume that at each leaf of the cover we have terms which, when applied to
a value, yield a reducible result, and need to deduce that the whole covered term
is itself reducible, that is, yields reducible result when applied to a value.
The first ingredient is the following lemma, which pushes application
to the leaves of a cover by virtue of commuting conversions.
We omit similar lemmas for the other eliminations.
% : projections, \(\rai\) and \(\iteop\).

\begin{lemma}[\formline{MetaTheory.v}{318}{app_cover} Covered application]
  \label{lem:push-app}
  If \(f \in \Tm(\Gamma,A\to B)\) is covered by \(F\) and \(u \in \Tm(\Gamma,A)\),
  then \(\red{f~u}{t}\) for some \(t\) which is covered by the predicate
  \[\Delta,t' \mapsto \exists (\rho : \Ren(\Delta,\Gamma)) (f' : \Tm(\Delta,A \to B)).
  (t' = f'~\subs{u}{\rho}) \wedge (F~f')\]
\end{lemma}

Because pushing argument to the leaves introduces extra variables
in the context, we must also know that values are preserved by context extension.
This is what necessitates the Kripke-style
quantification over renamings in the definition of values at function
types. For the category theorist, this should be
unsurprising: for values to form a presheaf over
renamings, the definition of the semantic function type mirrors
the presheaf exponential.

\begin{lemma}[\formline{MetaTheory.v}{235}{ren_value}]
  If \(v\) is a value of type \(A\) in \(\Gamma\) and \(\rho \in \Ren(\Delta,\Gamma)\),
  then \(\subs{v}{\rho}\) is a value.
\end{lemma}

Putting these together, we finally obtain:

\begin{lemma}[\formline{MetaTheory.v}{466}{covered_reducible} Reducible cover]
  \label{lem:red-cover}
  A term covered by reducible terms is reducible.
\end{lemma}

With this in hand, the rest of the proof is
history. We let interested reader consult the formalisation.
An interesting point to note, though, is that the analogue of \cref{lem:push-app}
for \(\rai\) and \(\iteop\) is not used to prove \cref{lem:red-cover}, but 
to show the semantic typing rule for
the eliminators: owing to the polarity of the types, the work for
negative types is for constructors --~where we need \cref{lem:red-cover}~--
while eliminations are the hard part for positive types.

\begin{theorem}[\formline{MetaTheory.v}{682}{fundamental} Fundamental lemma]
  If \(\typing{\Gamma}{t}[A]\) then \(\semty{\Gamma}{t}[A]\).
\end{theorem}

From \cref{lem:reflect} we get that variables are values at all types, and
thus that the identity substitution is a value substitution. Hence, we conclude

\begin{corollary}[\formline{MetaTheory.v}{718}{normalisation} Normalisation]
  Any well-typed term reduces to a normal form.
\end{corollary}

\section{Interpolation}
\label{sec:interpolation}

With our characterisation of normal forms in hand, we can move to the
interpolation result. Following \Cubric~\cite{Cubric1994}, we give a
polarised version in the fashion of Lyndon interpolation \cite{Lyndon1959interpolation},
% rather than the simpler Craig interpolation,
and thus start with a polarised notion of vocabulary.

\begin{definition}[\formline{Languages.v}{9}{atoms_ty} Vocabulary]
  Let \(A \in \Ty_{\mathcal{B}}\). We define \(A^+ \subseteq \mathcal{B}\)
    (resp.\ \(A^-\)), the \emph{positive}
    (resp.\ \emph{negative}) \emph{vocabulary} of \(A\) by induction,
    with $\overline{+}\coloneq -$ (resp. $\overline{-} \coloneq+$):\\[-1.5em]
  \begin{minipage}[t]{0.15\textwidth}
  \begin{align*}
    b^{+} & \coloneq \{b\} \\
    b^{-} & \coloneq \emptyset
  \end{align*}
  \end{minipage}
  \begin{minipage}[t]{0.15\textwidth}
  \begin{align*}
    \top^p & \coloneq \emptyset \\
    \bot^p & \coloneq \emptyset \\
  \end{align*}
  \end{minipage}
  \begin{minipage}[t]{0.25\textwidth}
  \begin{align*}
    (A + B)^{p} & \coloneq A^p \cup B^p \\
    (A \times B)^{p} & \coloneq A^p \cup B^p \\
  \end{align*}
  \end{minipage}
  \begin{minipage}[t]{0.2\textwidth}
  \begin{align*}
    (A \to B)^{p} & \coloneq A^{\overline{p}} \cup B^p \\
  \end{align*}
  \end{minipage}\\
This is lifted pointwise to contexts as
  \((\Gamma.(x : A))^p \coloneq \Gamma^p \cup A^p\).
\end{definition}

Although stating the final theorem does not need it, in the proof by induction
it is necessary to keep track of which variables in the context are ``coming from''
the conclusion, and thus we need to maintain a context partition in the following sense.

\begin{definition}[\formline{Interpolation.v}{22}{splits} Context partition]
  \emph{Context partition} \(\Gamma = \part{\Gamma_{1}}{\Gamma_{2}}\)
  is defined by
  \begin{mathpar}
    \inferdef{ }{\emp = \part{\emp}{\emp}} \and
    \inferdef{\Gamma = \part{\Gamma_{1}}{\Gamma_{2}}}{
      \Gamma. (x : A) = \part{\Gamma_{1}.(x : A)}{\Gamma_{2}}} \and
    \inferdef{\Gamma = \part{\Gamma_{1}}{\Gamma_{2}}}{
      \Gamma. (x : A) = \part{\Gamma_{1}}{\Gamma_{2}. (x : A)}}
  \end{mathpar}
\end{definition}

In a named setting this is a mere relation: as long as we enforce that variable
names in contexts are unique, they tell us how the partitioning is done.
In the formalisation, where we use a nameless syntax, context partition becomes
structure: there are two proofs of \(A.A = \part{A}{A}\), respectively corresponding,
in a named setting, to \((x:A).(y:A) = \part{(x:A)}{(y:A)}\) and
\((x:A).(y:A) = \part{(y:A)}{(x:A)}\). We spare the reader the substitution
calculus arising from this, \eg the fact that if \(\Gamma = \part{\Gamma_1}{\Gamma_2}\)
then there is a renaming \(\typing{\Gamma}{\rho}[\scol{\Gamma_1}]\).

\subsection{Constant-free theorem}

Our first theorem is constant-free: for
this section we fix a language \(\lang = (\base,\emptyset,\_)\).
It is expressed using an existential, but in the formalisation
we have a proper function (using \textsc{Equations} \cite{Sozeau2019})
computing the interpolating type and terms, which we prove correct.
This means we can nicely separate independent parts of the correctness
proof: the vocabulary restriction and equational constraint on the interpolating
terms are proven independent.

\begin{theorem}[\formline{Theorems.v}{79}{interpolation_empty}
    Interpolation, constant-free]
  Assume \(\typing[\lang]{\Gamma}{t}[T]\), then there exists a type \(M\) and terms
  \(\typing{\Gamma}{l}[M]\) and \(\typing{(x:M)}{r}[T]\) such that
  \begin{itemize}
  \item \(M\) uses only the common vocabulary of \(\Gamma\) and \(T\):
    for any \(p \in \{+,-\}\) \(M^p \subseteq \Gamma^p \cap T^p\);
  \item \(t\) is the composition of \(r\) and \(l\): \(\conv{\Gamma}{t}{\subs{r}{l..}}[T]\).
  \end{itemize}
\end{theorem}

Since by normalisation every term is equal to a normal form, it is enough
to show the following theorem for normal forms,
and specialize it to the trivial partition \(\Gamma = \part{\Gamma}{\emp}\).

\begin{theorem}[\formline{Interpolation.v}{352}{Interpolation}
    Interpolation, inductive statement]
    \label{thm:interpolation-ind}
  Assume a partitioned context \(\Gamma = \part{\Gamma_{s}}{\Gamma_{t}}\).
  Given any normal form \(\ccol{t} \in \Nf(\Gamma,\tcol{T})\),
  there is a type \(\mcol{M}\) and terms
  \(\scol{l} \in \Tm(\scol{\Gamma_{s}},\mcol{M})\) and
  \(\tcol{r} \in \Tm(\tcol{\Gamma_{t}}.\mcol{(x : M)}, \tcol{T})\)
  such that \(\conv{\Gamma}{\subs{\tcol{r}}{\scol{l}..}}{\ccol{t}}[\tcol{T}]\)
  and for all \(p\),
    \(\mcol{M}^p \subseteq \scol{\Gamma_{s}}^p \cap (\tcol{\Gamma_{t}}^{-p} \cup \tcol{T}^p)\).
  Moreover, given any neutral term \(\icol{t} \in \Ne(\Gamma,T)\), either
  \begin{enumerate}
    \item ``\(\scol{T}\) is a source type'', \ie \(\scol{T}^p \subseteq \scol{\Gamma_{s}}^p\),
      and we have a type \(\mcol{M}\) and terms \(\tcol{l} \in \Tm(\tcol{\Gamma_{t}},\mcol{M})\) and
  \(\scol{r} \in \Tm(\scol{\Gamma_{s}}.\mcol{(x : M)}, \scol{T})\) such that
  \(\conv{\Gamma}{\subs{\scol{r}}{\tcol{l}..}}{\icol{t}}[\scol{T}]\) and
  \(\mcol{M}^p \subseteq \scol{\Gamma_{s}}^{-p} \cap \tcol{\Gamma_{t}}^p\).

    \item or ``\(\tcol{T}\) is a target type'', \ie
      \(\tcol{T}^p \subseteq \tcol{\Gamma_{t}}^p\), and we have a type \(\mcol{M}\)
      and terms \(\scol{l} \in \Tm(\scol{\Gamma_{s}},\mcol{M})\) and
      \(\tcol{r} \in \Tm(\tcol{\Gamma_{t}}.\mcol{(x : M)}, \tcol{T})\) such that
       \(\conv{\Gamma}{\subs{\tcol{r}}{\scol{l}..}}{\icol{t}}[\scol{T}]\)
       and \(\mcol{M}^p \subseteq \scol{\Gamma_{s}}^p \cap \tcol{\Gamma_{t}}^{-p}\).
  \end{enumerate}
\end{theorem}

The statement for normal forms is essentially \Cubric's \cite[Lemma 3.7]{Cubric1994}.
It is a generalisation of
the interpolation statement, using the partition to account for some
type information initially in the target type moving in the context,
\eg in the case of a \(\lambda\)-abstraction.
The counterpart for neutrals, however, is new:
the original proof did not use neutrals at all, and instead, went by strong
induction on the size of terms, replacing subterms by variables.
Our reasoning by mutual induction is much more direct,
and also much easier to formalise!

Note that in case 1 the roles of the source and target context
are exchanged in the typing of the interpolants compared
to the statement for normal forms.
It might make sense to push this distinction to bidirectional typing,
and have a refined neutral judgment that also infers from which
side of the context split its main variable comes. This might slightly clarify
the present proof, but since it does not drastically affect it we will not
explore it further.

\begin{proof}
As expected, the proof goes by induction on the definition of \(\Ne/\Nf\).\\
  Most \(\Nf\) cases are rather direct, we only show 
  \nameref{rule:lam-b}, cases for \nameref{rule:pair-b}, \nameref{rule:inl-b}, \nameref{rule:inr-b} are similar. \\
  %
  % \nameref{rule:star-b}: \(\tcol{T}\) is \(\top\), so we can take
  % \(\mcol{M} = \top\), and \(\scol{l} = \tcol{r} = \star\).\\
  %
  % \nameref{rule:pair-b}: \(\tcol{T}\) is \(\tcol{T_1} \times \tcol{T_2}\), \(\ccol{t}\)
  % is \(\ccol{(t_1,t_2)}\). By induction,
  % we obtain \(\scol{l_i} \in \Tm(\scol{\Gamma_{s}},\mcol{M_i})\)
  % and \(\scol{r_i} \in \Tm(\tcol{\Gamma_{t}}.\mcol{(x : M_i)},\tcol{T_i})\).
  % We can then take \(\mcol{M} \coloneq \mcol{M_1 \times M_2}\) and
  % \[\begin{array}{rlll}
  %   \scol{\Gamma_{s}} \vdash & \scol{l} \coloneq &
  %     \scol{(l_1,l_2)} & \ty \mcol{M_1 \times M_2} \\
  %   \tcol{\Gamma_{t}}.\mcol{(x : M_1 \times M_2)} \vdash &
  %     \tcol{r} \coloneq & \tcol{(\subs{r_1}{\proj{1}{y}..},\subs{r_2}{\proj{2}{y}..})} &
  %     \ty \tcol{T_1 \times T_2}
  % \end{array}\]
  %
  \nameref{rule:lam-b}: We extend the \emph{target} context \(\tcol{\Gamma_{t}}\)
  with the extra variable (this is why we need the formulation with a context partition!)
  and by induction hypothesis, we obtain
  \(\scol{l'} \in \Tm(\scol{\Gamma_{s}},\mcol{M})\) and
  \(\tcol{r'} \in \Tm(\tcol{\Gamma_{t}}.\tcol{(x : A).\mcol{(z : M)}})\),
  and we can take
  \[\begin{array}{rlll}
  \scol{\Gamma_{s}} \vdash & \scol{l} \coloneq & \scol{l'} & \ty \mcol{M} \\
  \tcol{\Gamma_{t}}.\mcol{(z : M)} \vdash & \tcol{r} \coloneq &
    \tcol{\l x.r'} &
    \ty \tcol{A \to B}
  \end{array}\]
  %
  % \nameref{rule:pair-b}, \nameref{rule:inl-b}, \nameref{rule:inr-b} are similar:
  % we post-compose \(\tcol{r}\) as obtained by recursive interpolation
  % by the relevant constructor, except we do not need to think about context extension.

  Cases involving \(\Ne\) are more interesting. For each rule, we have two subcases,
  depending on whether the induction hypothesis for the main sub-neutral is
  in case 1 or 2 above, \ie whether the main variable of the neutral
  in is the source or target context.\\
  \nameref{rule:var-b}/source: We have \((x \ty \scol{A}) \in \scol{\Gamma_s}\).
  Then we take \(\mcol{M} \coloneq \top\)
  \[\begin{array}{rlll}
  \tcol{\Gamma_{t}} \vdash & \tcol{l} \coloneq & \tcol{\star} & \ty \mcol{M} \\
  \scol{\Gamma_{s}}.\mcol{(\_ : M)} \vdash & \scol{r} \coloneq &
    \scol{x} & \ty \scol{A}
  \end{array}\]
  \nameref{rule:var-b}/target: We have \((x \ty \tcol{A}) \in \tcol{\Gamma_t}\).
  Then we take \(\mcol{M} \coloneq \top\)
  \[\begin{array}{rlll}
  \scol{\Gamma_{s}} \vdash & \scol{l} \coloneq & \scol{\star} & \ty \mcol{M} \\
  \tcol{\Gamma_{t}}.\mcol{(\_ : M)} \vdash & \tcol{r} \coloneq &
    \tcol{x} & \ty \tcol{A}
  \end{array}\]
  \nameref{rule:demote-b}/source: By induction hypothesis, we get that \(\scol{A}\)
    is a source type,
    \(\tcol{l} \in \Tm(\tcol{\Gamma_{t}},\mcol{M'})\) and
    \(\scol{r} \in \Tm(\scol{\Gamma_{s}}.\mcol{(x : M')}, \scol{A})\).
    Here we need to handle the reversal between the source and target
    context appearing in case 1.
    Luckily, we have \(\scol{A} = \tcol{B}\),
    thus \(\mcol{A}\) can only use a vocabulary common to
    \(\scol{\Gamma_s}\) and \(\tcol{\Gamma_t}\).
    Hence, \(\mcol{M} \coloneq \mcol{M' \to A}\) is a valid interpolating type,
    and we take
    \[\begin{array}{rlll}
      \scol{\Gamma_{s}} \vdash & \scol{l'} \coloneq & \scol{\l x.r} & \ty \mcol{M' \to A} \\
      \tcol{\Gamma_{t}}.\mcol{(z : M' \to A)} \vdash & \tcol{r'} \coloneq &
        \tcol{z~l} & \ty \tcol{A}
    \end{array}\]
  \\
  \nameref{rule:demote-b}/target: In this case we keep the exact
    same interpolating type and terms.\\
  \nameref{rule:raise-b}, \nameref{rule:if-b}: similarly to \nameref{rule:demote-b},
    we need to handle the reversal of source and target context.\\
  \nameref{rule:app-b}/source: We consider
  \(\icol{t}~\ccol{u}\), and the induction hypothesis on \(\icol{t}\) gives
  some \(\scol{A \to B}\) with vocabulary in \(\scol{\Gamma_s}\) and
  \(\mcol{M_t}\), \(\scol{r_t}\) and \(\tcol{l_t}\) such that
  \(\icol{t} = \subs{\scol{r_t}}{\tcol{l_t}..}\). Since \(\scol{A}\)'s
  vocabulary is in \(\scol{\Gamma_s}\), we can
  interpolate \(\ccol{u}\) with a reversed colouring,
  \ie inverting the roles of \(\scol{\Gamma_s}\) and \(\tcol{\Gamma_t}\).
  This gives \(\mcol{M_u}\), \(\scol{r_u}\) and \(\tcol{l_u}\) such that
  \(\ccol{u} = \subs{\scol{r_u}}{\tcol{l_u}..}\). We take
  \(\mcol{M} \coloneq \mcol{M_t} \times \mcol{M_u}\), and
  \[\begin{array}{rlll}
      \tcol{\Gamma_{t}} \vdash & \tcol{l} \coloneq & \tcol{(l_t,l_u)} &
        \ty \mcol{M_t \times M_u} \\
      \scol{\Gamma_{s}}.\mcol{(z : M_t \times M_u)} \vdash & \scol{r} \coloneq &
        \scol{\subs{r_t}{\proj{1}{z}..}~\subs{r_u}{\proj{2}{z}..}} & \ty \scol{B}
  \end{array}\]
  \nameref{rule:app-b}/target: This time, the induction hypothesis on \(\icol{t}\) gives
  us some \(\tcol{A \to B}\) with vocabulary in \(\tcol{\Gamma_s}\) and
  \(\mcol{M_t}\), \(\tcol{r_t}\) and \(\scol{l_t}\) such that
  \(\icol{t} = \subs{\tcol{r_t}}{\scol{l_t}..}\). We can directly use the
  induction hypothesis on \(\ccol{u}\), to get
  \(\mcol{M_u}\), \(\tcol{r_u}\) and \(\scol{l_u}\) such that
  \(\ccol{u} = \subs{\tcol{r_u}}{\scol{l_u}..}\). Then we take again
  \(\mcol{M} \coloneq \mcol{M_t} \times \mcol{M_u}\), and
  \[\begin{array}{rlll}
      \scol{\Gamma_{s}} \vdash & \scol{l} \coloneq & \scol{(l_t,l_u)} &
        \ty \mcol{M_t \times M_u} \\
      \tcol{\Gamma_{t}}.\mcol{(z : M_t \times M_u)} \vdash & \tcol{r} \coloneq &
        \tcol{\subs{r_t}{\proj{1}{z}..}~\subs{r_u}{\proj{2}{z}..}} & \ty \tcol{B}
  \end{array}\]
  \nameref{rule:proji-b}: Here in all cases we keep the same interpolating type
    and left interpolant, and simply post-compose the right interpolant with
    the relevant projection.
\end{proof}

\begin{remark}
The interpolant we construct is in fact not
the same as \Cubric. Indeed, we only build
an arrow type as the end of a ``neutral phase'' (in case 
\nameref{rule:demote-b}/source), but use a product type to
gather the type of the arguments. Instead, \Cubric will build an iterated
function type in a similar context, meaning his interpolant is a currification of
ours. For example, consider
\(\check{(\part{\Gamma_{s}}{\Gamma_{t}})}{\icol{x}~\ccol{a}~\ccol{b}}{\tcol{C}}\)
with \((x : A \to B \to C) \in \scol{\Gamma_{s}}\). The interpolants
we build are as follows (where \(\mcol{M_a}\), \(\scol{r_a}\), \(\tcol{l_a}\),
come from interpolating \(\ccol{a}\) and \(\mcol{M_b}\), \(\scol{r_b}\), \(\tcol{l_b}\)
from \(\ccol{b}\))
\[\begin{array}{rcl}
  \scol{\Gamma_{s}} \quad \vdash
    & \scol{\l y.~x~\subs{r_a}{(\proj{2}{\proj{1}{y}})..} ~ \subs{r_b}{\proj{2}{y}..}}
    & \ty \quad \mcol{(\top \times M_a \times M_b) \to C} \\
  \tcol{\Gamma_{t}}.~\mcol{z :(\top \times M_a \times M_b) \to C} \quad \vdash
    & \tcol{z~(\st,l_b,l_a)}
    & \ty \qquad \tcol{T}
  \end{array}\]
Instead, \Cubric's interpolants are (with the same notations)
\[\begin{array}{rcl}
  \scol{\Gamma_{s}} \quad \vdash
    & \scol{\l z_a~z_b.~x~\subs{r_a}{z_a..} ~ \subs{r_b}{z_b..}}
    & \ty \quad \mcol{M_a \to M_b \to C} \\
  \tcol{\Gamma_{t}}.~\mcol{z : M_a \to M_b \to C} \quad \vdash
    & \tcol{z~l_a~l_b}
    & \ty \qquad \tcol{T}
  \end{array}\]
\end{remark}

\subsection{Constants and equations}

We can generalise the above to arbitrary languages \(\lang = (\base,\const,\tyof)\)
with arbitrary equations.
As explained in \cref{sec:red-eq}, handling the extra equations is essentially
trivial. As for constants, we have to be a bit careful that
the interpolant might also use vocabulary present in the constants. The
final theorem is the following:

\begin{theorem}[\formline{Theorems.v}{54}{interpolation_eq}
    Interpolation, with constants]
  Assume \(\typing[\lang]{\Gamma}{t}[T]\), and a partition
  \(\const = \const_{s} \cup \const_{t}\) of constants.
  Define \(\const_{i}^p \coloneq \cup_{c \in \const_{i}} (\tyof c)^p\)
    the vocabulary of the sets of constants, and
    \(\lang_{i} \coloneq (\base,\const_{i},\tyof_{i})\)
    where \(\tyof_{i}\) is the obvious restriction.
  Then there exists a type \(M\) and terms
  \(\typing[\lang_{l}]{\Gamma}{l}[M]\) and \(\typing[\lang_{r}]{(x:M)}{r}[T]\) such that
  \begin{itemize}
  \item \(M\) uses only the vocabulary common to \(\Gamma\) and \(\const_{s}\)
    on one side, and \(T\) and \(\const_{t}\) on the other:
    for any \(p \in \{+,-\}\)
    \(M^p \subseteq (\Gamma^p \cup \const_{s}^p) \cap (T^p \cup \const_{t}^p)\);
  \item \(t\) is the composition of \(r\) and \(l\): \(\conv{\Gamma}{t}{\subs{r}{l..}}[T]\).
  \end{itemize}
\end{theorem}

\section{On formalisation choices: highs and lows of automation}
\label{sec:formalisation}

As often in this sort of work, proofs are relatively straightforward
once the right organisation --~\eg induction statements~-- have been found.
We try to reproduce this using \Rocq's
automation. We found solvers for well-specified classes of problems particularly
valuable, typically for goals only
mentioned in passing on paper but which formalisation cannot paper over.
Yet, we can also testify
that, if solvers are not well polished, a lot of their value evaporates:
a slow or, worse, unreliable solver is not much better than no solver at all,
because it loses its main appeal, which is to stop thinking
about boring goals altogether.

\subparagraph*{Substitutions}
 We use \textsc{AutoSubst2-OCaml}
\cite{Stark2019,Dapprich2021} to handle the substitution calculus, in particular
the solver \texttt{asimpl} for equations between substitutions.
This saves us much boring work, particularly when reasoning about
proving that \(\red{\subs{r}{l..}}{t}\) in \cref{thm:interpolation-ind}, which
involves rather nasty terms. 
Although very useful, the tool still has a few caveats. First and foremost,
\texttt{asimpl} is actually incomplete, because it is missing the equation
\(((\wk \circ \sigma),\sigma(x_0)) = \sigma\), where \(x_0\) is the last
variable in context and \(\circ\) is composition of substitutions.
Concretely, \texttt{asimpl} does not recognise identification
that need a case-split on whether the argument is the first
variable, which often arises around context extension.
Other sources of friction include speed, with the tactic
quickly taking multiple seconds, and a perfectible setup around notations
and abstraction barriers.

The issues of speed and incompleteness are addressed by the
\textsc{Sulfur} plugin \cite{BouverotDupuis2026}, a spiritual successor of
\textsc{Autosubst}. Sadly, we were not able to use it
due to the way we parameterised all definitions by a type-class
representing a language. With \textsc{Autosubst},
that parameterisation was a matter of changing a few lines in the
generated file. Although not very principled, this was easy. In contrast,
the plugin approach does not expose the generated code,
which we thus cannot tweak. Another issue
is that \textsc{Sulfur} currently supports a single type of terms, while
we use \textsc{Autosubst} to also deal with substitutions for eliminations.

As more long-term goals, we also had frictions with lemmas not applying
because their conclusions only unify with the goal up to
substitution equations. Integrating matching modulo the substitution calculus would
thus be a very interesting extension. Additionally, an aspect of the
confluence proof we were afraid to formalise is the proliferation of relations
--~\(\beta\)-reduction, commuting conversion, parallel reduction, etc.~--
which are all congruences. Indeed, boilerplate reasoning about congruence
closures is already tedious in the current reasoning about reduction and
conversion, we were afraid to simply drown in extremely uninteresting lemmas.
It seems possible to extend \textsc{Autosubst}/\textsc{Sulfur} to
deal with such issues of congruence closure, which would remove a great amount of tedium.

% As an alternative, we should mention \textsc{Pyrosome} \cite{Jamner2025}, which
% also automates syntax-related reasoning, directly at the level 
% of typed syntax (specified as a generalised algebraic theory
% \cite{Cartmell1986,Kaposi2019a}). While this has the advantage of
% also handling some typing-related boilerplate, it does not come
% with a decision procedure, an extremely valuable part of \textsc{Autosubst}.

Finally, more specific to our setting is the need to reason about context splitting.
We do this by maintaining weakenings which relate the split contexts
to the original one, which add noise to the equational reasoning.
These issues are typical when working with linear languages, where
every typing rule has to deal with context splitting. Some general
infrastructure, based on masks, has been developed to
deal with this in formalisation \cite{Zackon2026}. It would be interesting to
explore whether it would simplify our proofs.

\subparagraph*{Solvers and setoid rewriting}

Beyond \textsc{Autosubst}, we heavily use \Rocq's proof search mechanisms.
We store all typing-related lemmas in a hint database, with which we systematically discharge typing goals. Owing to the simplicity of \stlc's typing,
this is very efficient. For instance, the proofs of progress and preservation
(\formline{MetaTheory.v}{6}{TypingRed}) for the whole language take 30 lines.
We also use \textsc{Std\(++\)}'s set theory library,
and its solver, which once again makes an entire class
of boring goals disappear. Sadly, here too we encountered performance issues.

To handle reduction, we exploit setoid rewriting to apply reduction steps or
hypotheses deep in terms. This proved quite useful, although less automated
than what \texttt{asimpl} provides. It seemed possible to also automate this
sort of goals by implementing a (reflexive) tactic that computes the normal
form of a term with respect to reduction, but the amount of work did not seem
to be worth it, so we kept to manual rewriting. An interesting difficulty appeared
for conversion: \Rocq's current setoid rewriting
machinery is not designed to work for non-reflexive relations, such as conversion,
which is only reflexive on well-typed terms. However, by adding the right
\texttt{Hint Extern} (\formline{Equations.v}{238}{Hint})
that applies our automation for typing to resolve reflexivity
goals, we managed to use setoid rewriting for conversion. This suggests it should
be more generally feasible for setoid rewriting to work well also for non-reflexive
relations, at the cost of setting up more type-class search to discharge
reflexive goals.

\subparagraph*{Confluence}

As already mentioned, we did not formalise \cref{thm:confluence}.
Although confluence is often considered simpler than normalisation,
its formal proof seemed daunting. First,
confluence of \(\beta\) and commutation of \(\beta\) and commuting
conversion, while conceptually simple, involve the definition of two new relations,
each with a congruence closure, leading to goal proliferation.
Even more frightening is local confluence of commuting conversion, where
\textsc{CSI\textsuperscript{ho}} reports 12 critical pairs. While each
is simple enough, the full argument, both the reduction of
local confluence to the pair's closures and the closure itself, seemed
extremely tedious. It would be interesting to explore what could be done here,
in the form of libraries for congruence closures or automation, maybe by
having \textsc{CSI\textsuperscript{ho}} generate proof witnesses.

\section{Related work}
\label{sec:related}

% We discuss some of the most related works, about proof-relevant interpolation, formalised interpolation results, normalisation properties with sums as well as related works on bidirectional typing. 

\subparagraph*{(Relevant) Interpolation} 
First, it is important to stress that \Cubric's work is closely related to the original proof of Craig's interpolation theorem for natural deduction by Prawitz~\cite{Prawitz1965-PRANDA}: the construction of the interpolant in Prawitz' work is actually compatible with \Cubric's result, even though this is best seen in the $\lambda$-calculus. 
\Cubric's work did not have many extensions.
Matthes' \cite{Matthes-interpolation-01} work is a notable such development, that extends \Cubric's to a $\lambda$-calculus with generalised elimination. More recently, Fiorillo~\cite{Fiorillo2024} extended \Cubric's work to Parigot's $\lambda\mu$-calculus, allowing him to get rid of the sum type to express the interpolating type. It would indeed be natural to consider extending our formalisation to such settings.

In other proof formalism, Saurin analyses (proof-relevant) interpolation~\cite{Saurin25fscd} in sequent calculi for various logics
% (classical, intuitionistic and linear sequent calculi)
as a cut-introduction mechanism. This work also presents a ``computational'' version of interpolation in System L \cite{Curien2000,MunchMaccagnoni2013},
a term language for the sequent calculus.
Veltri and Wan studied proof-relevant interpolation for a non-associative restriction of the Lambek calculus~\cite{VeltriW25} (also formalised, see below), while
Fiorillo \textit{et al.}~\cite{Fiorillo2025} considered
linear logic proof nets or non-cut-free proofs.

While Craig interpolation received both syntactic and semantic proofs, it is natural to wonder what would be a semantic counterpart to proof-relevant interpolation: this was the original focus of \Cubric who studied interpolation in bicartesian-closed categories~\cite{Cubricphd,Cubric1994}; Fiorillo~\cite{Fiorillo2024} investigated such results in models of linear logic.

\subparagraph*{Formalised Interpolation}
Craig interpolation was formalised multiple times, both for classical~\cite{Ridge2006,michaelis_et_al2017} and intuitionistic logic~\cite{ChapmanMU08}.%
\footnote{Chapman {\it et al.}\ as well as Ridge cite Boulmé~\cite{Boulme1996}
as an even earlier formalisation for classical logic in \Rocq,
but this development does not appear to be publicly available.}
More recently, Férée \textit{et. al.}
have formalised uniform interpolation results {\it à la} Pitts~\cite{Pitts92}
in \Rocq,
first for intuitionistic logic \cite{Feree2023}, then for modal logics
\cite{Feree2024}. Although related, uniform interpolation is different enough
from Craig interpolation that their proofs are quite different from ours. In particular,
they formalise cut admissibility but no normalisation, and again there
is no proof-relevant refinement. Even more recently, Veltri and Wan~\cite{VeltriW25}
used \textsc{Agda} to formalise cut-elimination and proof-relevant interpolation 
for a substructural logic, the Lambek calculus, in their effort to understand this system.
Their propositional logic only
has negative connectives, but the difficulty relies on the complexity of
the context's structure due to the lack of structural rules:
Working with a non-associative, non-commutative sequent calculus, they develop a complex machinery to handle sequents with a tree structure.
While we formalized a proof-relevant version of Prawitz's approach,
they looked at a proof-relevant version of Maehara \cite{Maehara}. 

Compared to all these, our proof (and formalisation) is by far the most
``\(\lambda\)-calculesque'', owing to its emphasis on proof-relevance.
It is also rather compact (\char`~4kLoC, including 800LoC generated by \textsc{Autosubst})
% almost 10 times smaller than Férée \textit{et al.} \cite{Feree2024}),
and fast (\char`~30s on a standard laptop,
most of which is paid to slow solvers, to compare to Veltri and Wan's development,
which takes hours to compile), so we hope some of our formalisation techniques
can inspire others.

\subparagraph*{Normalisation for sums}
Since Prawitz's work \cite{Prawitz1965-PRANDA},
normalisation/cut-elimination for sums has been the subject of much
focus. The bad behaviour of sums in natural deduction is in fact
one of the major reasons to move to sequent calculi.
Staying in the context of the \(\lambda\)-calculus (and thus natural
deduction), the closest work to ours is that of Abel and Sattler~\cite{Abel2019a}.
Their normalisation result is formalised in \textsc{Agda}, although in a
different, more intrinsic style. They merely give a
normalisation function, but leave as future work the proof that the
result of this function relates to its input, which is where a logical
relation argument such as ours is necessary. Thus, to the best of our
knowledge ours is the first formalised proof of normalisation for \stlcp.

As explained in \cref{sec:intro}, we are only concerned with a ``weak''
kind of normal forms, which have the subformula property but do
not give canonical representatives of equivalence classes for the
equational theory of biCCC. For that, one needs much more involved
constructions beyond the scope of this paper, but a good
account of which can be found in Scherer~\cite{Scherer2017}.

\subparagraph*{Exploiting bidirectional typing}
Although we ended up rediscovering some of it,
the relationship between bidirectional typing and the subformula property
is not new: it is mentioned in Dunfield and
Krishnaswami's~\cite{Dunfield2021} survey, which credits it to Pfenning,
and indeed the connection has appeared somewhat implicitly in his lectures
for a long time \cite{Pfenning1999}.
The relation between normal forms and the subformula property is also
explicitly mentioned by Abel and Sattler~\cite{Abel2019a}, although without
connection to bidirectional typing. We should also mention the work of
Lennon-Bertrand around \textsc{MetaRocq}
\cite{LennonBertrand2021,MetaCoq2025}, which implicitly exploits the subformula
property to prove strengthening for a dependent typing judgment.
But altogether, although the connection was ``in the air''
we seem to be the first to consciously exploit it to this degree.

Given that sequent calculus is generally a better setting for interpolation,
including proof-relevant results \cite{Saurin25fscd}, it seems tempting to
export our methodology. Luckily, Mihejevs and Hedges \cite{Mihejevs2025}
recently studied bidirectional typing for System L \cite{MunchMaccagnoni2013,Curien2000},
a syntax for sequent calculi, showing there is a remarkable coincidence between
the bidirectional discipline and System L's structure.
It seems enticing to explore the connections between these.

\section{Conclusion}
\label{sec:conclusion}

\paragraph*{Summary}

In the present paper, we have revisited the seminal work by \Cubric
on proof-relevant interpolation in the $\lambda$-calculus with sums (\stlcp)
\cite{Cubric1994}. Our contribution is twofold:
\begin{itemize}
	\item We revisited \Cubric' proof (which followed closely Prawitz' proof of Craig's interpolation theorem) and simplified its inductive argument by developing a bidirectional type system characterising the normal forms of \stlcp, with  $\beta$ and commuting conversions. 
	\item We benefited from this restructured proof to allow for a \Rocq formalisation.
  In doing so, we developed the meta-theory of \stlcp including normalisation.
\end{itemize} 

\paragraph*{Future work}

We envision three particularly interesting directions for future work, corresponding to strengthening our results in orthogonal direction: (i) interpolating non-normal terms, (ii) proof-relevant uniform interpolation and (iii) extending \Cubric's results to dependent types. 

\subparagraph*{Extensions with cuts}
The first evident direction for future work consists in lifting the main issue of
\Cubric interpolation: it only applies to normal forms. From a computational perspective,
this is indeed a drastic limitation: every computation can be factored through an interface type,
but only by first normalising the term, an extremely costly operation!

What if the term is not in normal form? 
Considering that interpolation in fact requires a much weaker condition than
cut-freeness, one can hope to still interpolate terms as long as their
cuts satisfy some assumptions. 
Fiorillo {\it et al.}~\cite{Fiorillo2025} have results in this direction,
extending Saurin's \cite{Saurin25fscd}
results to proofs with cuts, as long as cut formulae satisfy vocabulary constraints.
We believe this should adapt to \stlcp, and see at least two paths to prove it.

In their work, Fiorillo {\it et al.} identify a quite general condition on cut formulae
to carry out proof-relevant interpolation in a sequent calculus with cuts, and show this can be achieved for MELL (multiplicative-exponential linear logic) proof nets with cuts too (under similar assumptions). Since \stlcp can be encoded in the intuitionistic fragment of MELL using linear embeddings of intuitionistic logic, one can hope to establish the result on $\lambda$-terms with redexes via a translation to proof nets. The main difficulty to achieve such a result along these lines amounts to designing a faithful embedding of the whole of \stlcp in intuitionistic MELL that satisfies a good simulation property.

In the light of the present paper, another strategy is to amend the bidirectional type system to characterise terms whose redexes satisfy the assumptions of the interpolation theorem with cuts. It is known that, to make
bidirectional typing go beyond normal form, the route to go is
to add annotations in terms \cite{McBride2022,Gratzer2019}, with a rule like the following:
\begin{mathpar}
  \inferdef[Annot]{\check{\Gamma}{t}{A}}{\infer{\Gamma}{\annot{t}{A}}{A}}
\end{mathpar}
In such a system, any redex is mediated by an annotation, for instance a (function)
\(\beta\)-redex must be of the form \((\annot{\l x.t}{A \to B})~\ccol{u}\). Thus,
the annotation exactly makes explicit what the cut formula is! With all cut formulas
visible this way, it should be possible to impose the same sort of
constraints on vocabulary as Fiorillo \textit{et al.}

\subparagraph*{Uniform interpolation}
A very natural and exciting question is the following: can the proof-relevant interpolation result presented above be ``uniformised''? 
While the current result ensures that given types $A$ and $B$ and a normal term
$\check{x: A}{t}{B}$, there exists a type $I$ satisfying the vocabulary conditions and two terms $u,v$ \emph{depending on} $A,B$ and $\ccol{t}$, such that $\red{\subs{u}{v..}}{\ccol{t}}$, a uniform proof relevant interpolation could be stated in two ways
(where we write \(\voc(A)\) for the vocabulary of \(A\), barring polarity
concerns):
\begin{description}
\item[Uniform] For every type $A$ and $L \subseteq \base$ there exists $I$ such that $\voc(I)\subseteq L$ and for \emph{any} $B$ satisfying $\voc(A) \cap \voc(B) \subseteq L$,
and \emph{any} $\check{x:A}{t}{B}$, there are terms $u,v$ such that
  $\typing{x: A}{v}[I]$ and $\typing{y: I}{u}[B]$ and $ \red{\subs{u}{v..}}{t}$.
\item[Strongly uniform] For every type $A$ and $L \subseteq \base$ there exists $I$ such that
  $\voc(I)\subseteq L$ and $\typing{x: A}{v}[I]$ such that for any $B$ satisfying $\voc(A)\cap \voc(B) \subseteq L$ and $\typing{x:A}{t}[B]$, there is a $\typing{y: I}{u}[B]$
    such that $ \red{\subs{u}{v..}}{t}$.
\end{description}
Given that such (strong) uniform interpolation results exist in the non-relevant
case, it seems plausible that they could yield a proof-relevant version. Moreover,
our formalisation seems like a good place to start, as it clarifies the structure
of \Cubric's proof and makes it easy to play around with the interpolant, which
is explicitly computed.

\subparagraph*{Dependent types} 
While Craig interpolation is most interesting in the setting of first-order logic,
\Cubric proof-relevant interpolation only deals with a simple and,
in particular, propositional type theory. 
We plan to investigate in the future how and whether \Cubric's theorem
can be extended to various forms of type dependency. Given the intrication of
types and terms in that context, a good version of proof-relevant interpolation
is a nice place to start. Here one evident difficulty is that, contrarily
to the simply-typed case, commuting conversions are not in general well-typed
in dependent type theory. A typical example is the so-called ``convoy pattern'',
in \Rocq syntax:\\
\begin{minipage}{\linewidth}
\lstset{backgroundcolor=\color{white}}
\begin{coq}
Definition head {A n} (x : Vector.t A (S n)) : A :=
  match x in Vector.t _ m return (m = S n -> A) with
    | nil _ => fun (e : 0 = S n) => False_rect _ (O_S _ e)
    | cons _ a _ _ => fun _ => a
  end eq_refl.
\end{coq}
\end{minipage}\\
Here pushing the application inside the branch is not well-typed, because
the type of the branches, which is given by the \texttt{return} clause,
is not the same as that of the entire match: the first branch has type
\texttt{(0 = S n) -> A} and applying this function to \texttt{eq\_refl} is
ill-typed.

On the other hand, dependent type theory is very expressive, and just like in
System F one obtains trivial uniform interpolants by using the
\(\forall\)/\(\exists\) quantifiers, so in dependent type theory
too naive a statement would be trivial to satisfy thanks to
\(\P\) and \(\Sigma\) types.

\paragraph*{Large Language Model Use}
No large language model or generative artificial intelligence tool was used in
developing the formalisation or writing the article.

%%
%% Bibliography
%%

%% Please use bibtex, 

\bibliography{biblio}

\end{document}